\documentclass{amsart}
\usepackage{amssymb}
\usepackage[maxnames=4,giveninits=true,url=false,doi=false,isbn=false,citestyle=numeric]{biblatex}
\usepackage{booktabs}
\usepackage{calc}
\usepackage{enumitem}
\usepackage{xcolor}
\usepackage{hyperref}
\usepackage{microtype}
\usepackage{psfrag}
\usepackage{ifpdf}
\ifpdf\usepackage{pst-pdf}\else\fi

\setlist[enumerate]{label=\textnormal{\bfseries(\alph*)}, leftmargin=*, nosep, widest=a}

\let\oldautoref\autoref
\makeatletter
\renewcommand\autoref[1]{\@first@ref#1,@}
\def\@throw@dot#1.#2@{#1}
\def\@set@refname#1{
    \edef\@tmp{\getrefbykeydefault{#1}{anchor}{}}%
    \def\@refname{\@nameuse{\expandafter\@throw@dot\@tmp.@autorefname}s}%
}
\def\@first@ref#1,#2{%
  \ifx#2@\oldautoref{#1}\let\@secondref\@gobble
  \else%
    \@set@refname{#1}
    \@refname~\ref{#1}
    \let\@secondref\@second@ref
  \fi%
  \@secondref#2%
}
\def\@second@ref#1,#2{%
  \ifx#2@ and~\ref{#1}\let\@nextref\@gobble
  \else, \ref{#1}
    \let\@nextref\@next@ref
  \fi%
  \@nextref#2%
}
\def\@next@ref#1,#2{%
   \ifx#2@, and~\ref{#1}\let\@nextref\@gobble
   \else, \ref{#1}
   \fi%
   \@nextref#2%
}
\makeatother

\def\equationautorefname~#1\null{%
  Equation~(#1)\null
}

\theoremstyle{plain}

\newtheorem{theorem}{Theorem}

\newtheorem{lemma}{Lemma}

\newtheorem{remark}{Remark}

\theoremstyle{definition}

\newtheorem{definition}{Definition}

\allowdisplaybreaks

\newcommand{\smashedtilde}[1]{\vphantom{#1}\smash{\tilde{#1}}}

\begin{document}

\title[Numerical approximations of fractional Brownian motion]{Strong convergence rates for Markovian representations of fractional processes}
\date{\today} 

\author{Philipp Harms}
\address{Department of Stochastics\\ University of Freiburg}
\email{philipp.harms@stochastik.uni-freiburg.de}
\thanks{The author gratefully acknowledges support in the form of a Junior Fellowship of the Freiburg Institute of Advances Studies.}

\begin{abstract}
Many fractional processes can be represented as an integral over a family of Ornstein--Uhlenbeck processes. 
This representation naturally lends itself to numerical discretizations, which are shown in this paper to have strong convergence rates of arbitrarily high polynomial order. 
This explains the potential, but also some limitations of such representations as the basis of Monte Carlo schemes for fractional volatility models such as the rough Bergomi model.
\end{abstract}

\maketitle

\tableofcontents

\section{Introduction}

This paper establishes strong convergence rates for certain numerical approximations of fractional processes. 
These approximations are inspired by Markovian representations of fractional Brownian motion \cite{carmona1998fractional,  carmona2000approximation, muravlev2011representation, harms2019affine} and of more general Volterra processes with singular kernels \cite{mytnik2015uniqueness, abijaber2019affine, abijaber2019markovian, abijaber2019multifactor, cuchiero2020generalized}.
The simplest such representation takes the form
\begin{align*}
W^H_t 
&:= 
\int_0^t (t-s)^{H-1/2} dW_s
=
\int_0^\infty \underbrace{\frac{1}{\Gamma(\tfrac12-H)}\int_0^t e^{(t-s)x} dW_s}_{=:Y_t(x)} \frac{dx}{x^{H+1/2}}, 
&&
t \in [0,\infty),
\end{align*}
where $W$ is standard Brownian motion, $W^H$ is Volterra Brownian motion\footnote{Also known as Riemann--Liouville fractional Brownian motion or L\'evy's definition of fractional Brownian motion.} with Hurst index $H \in (0,1/2)$, and $Y(x)$ is an Ornstein--Uhlenbeck process with speed of mean reversion $x \in (0,\infty)$.
The random field $Y_t(x)$, which is depicted in \autoref{fig:random_field}, has a version which is H\"older continuous in $t$ and smooth in $x$; see \autoref{lem:vol} for the precise statement.
Thanks to this spatial smoothness, the integral $dx$ can be approximated efficiently using high-order quadrature rules, following and extending \cite{carmona2000approximation, harms2019affine, abijaber2019lifting, abijaber2019multifactor}.
This leads to numerical approximations of the Volterra Brownian motion $W^H$.

\begin{figure}[h]%
\centering
\tiny
\begin{psfrags}%
\psfrag{a00}[Bc][Bc]{\PFGstyle $0.0$}%
\psfrag{a02}[Bc][Bc]{\PFGstyle $0.2$}%
\psfrag{a04}[Bc][Bc]{\PFGstyle $0.4$}%
\psfrag{a06}[Bc][Bc]{\PFGstyle $0.6$}%
\psfrag{a08}[Bc][Bc]{\PFGstyle $0.8$}%
\psfrag{a0}{\PFGstyle $\text{ 0}$}%
\psfrag{a10A}{\PFGstyle $10$}%
\psfrag{a10}[Bc][Bc]{\PFGstyle $1.0$}%
\psfrag{a5}{\PFGstyle $\text{ 5}$}%
\psfrag{t}[Bc][Bc]{\PFGstyle $\text{t}$}%
\psfrag{x0}[tc][tc]{\PFGstyle $0$}%
\psfrag{x11}[tc][tc]{\PFGstyle $1$}%
\psfrag{x2}[tc][tc]{\PFGstyle $0.2$}%
\psfrag{x4}[tc][tc]{\PFGstyle $0.4$}%
\psfrag{x6}[tc][tc]{\PFGstyle $0.6$}%
\psfrag{x8}[tc][tc]{\PFGstyle $0.8$}%
\psfrag{x}[Bc][Bc]{\PFGstyle $\text{x}$}%
\psfrag{y0}[cr][cr]{\PFGstyle $0$}%
\psfrag{y11}[cr][cr]{\PFGstyle $1$}%
\psfrag{y2}[cr][cr]{\PFGstyle $0.2$}%
\psfrag{y4}[cr][cr]{\PFGstyle $0.4$}%
\psfrag{y6}[cr][cr]{\PFGstyle $0.6$}%
\psfrag{y8}[cr][cr]{\PFGstyle $0.8$}%
\psfrag{Yxt}[Bc][Bc]{\PFGstyle $Y^x_t$}%
\psfrag{AAA}[tc][tc]{\PFGstyle $0.2$}%
\includegraphics[width=0.7\textwidth]{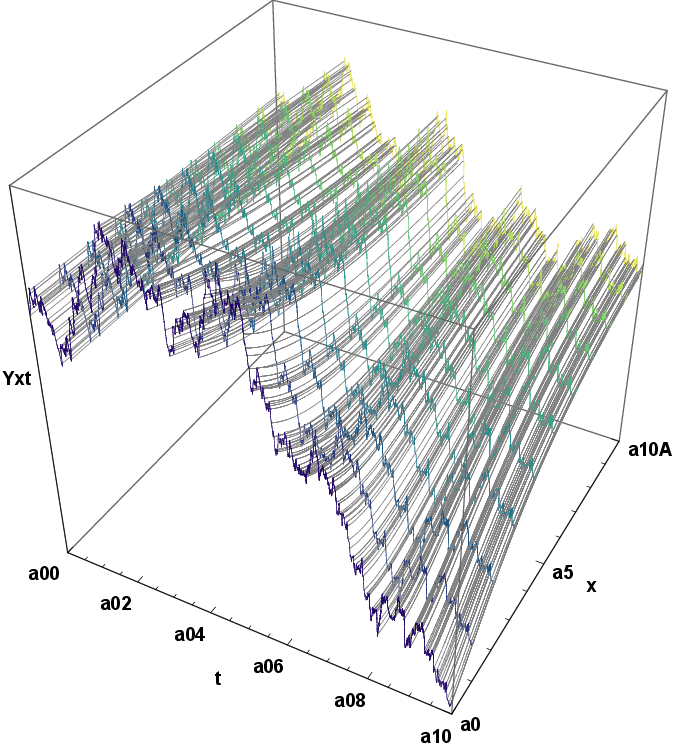}
\end{psfrags}
\caption{Volterra Brownian motion of Hurst index $H\in(0,1/2)$ can be represented as an integral $W^H_t=\int_0^\infty Y_t(x) x^{-1/2-H}dx$ over a Gaussian random field $Y_t(x)$. The smoothness of the random field in the spatial dimension $x$ allows one to approximate this integral efficiently using high order quadrature rules.}
\label{fig:random_field}
\end{figure}

The main result of this article is that Volterra Brownian motion can be approximated at arbitrarily high polynomial convergence rates by weighted sums of Ornstein--Uhlenbeck processes; see \autoref{thm:main} for the precise statement and error criterion.
By arbitrarily high polynomial convergence rates we mean that $m$-point interpolatory quadrature on $n$ suitably chosen spatial quadrature intervals leads to a discretization error of order $n^{-r}$ for all $r\in (0,2Hm/3)$; see \autoref{rem:rate}.
Thus, a given rate $r>0$ can be achieved by choosing $m>3r/(2H)$.
Note that low Hurst indices $H$ require high spatial quadrature orders $m$ to achieve a given approximation rate $r$. 
A visual impression of the quality of this approximation can be obtained from \autoref{fig:paths}.
The upper bound $2Hm/3$ on the convergence rate closely matches the numerically observed rate; see \autoref{fig:errors}. 

\begin{figure}[h]%
\begin{minipage}{0.5\textwidth}
\tiny
\begin{psfrags}%
\psfrag{AAA}[tc][tc]{\PFGstyle $0.2$}%
\psfrag{AAB}[tc][tc]{\PFGstyle $0.4$}%
\psfrag{AAC}[tc][tc]{\PFGstyle $0.6$}%
\psfrag{AAD}[tc][tc]{\PFGstyle $0.8$}%
\psfrag{AAE}[tc][tc]{\PFGstyle $1.0$}%
\psfrag{AA}[tc][tc]{\PFGstyle $0.0$}%
\psfrag{B}[tc][tc]{\PFGstyle $W^{H,n}_t$}%
\psfrag{TimeT}[bc][Bc]{\PFGstyle $\text{time ($t$)}$}%
\psfrag{VVVA}[cr][cr]{\PFGstyle $-2$}%
\psfrag{VVVB}[cr][cr]{\PFGstyle $0$}%
\psfrag{VVVC}[cr][cr]{\PFGstyle $2$}%
\psfrag{VVV}[cr][cr]{\PFGstyle $-4$}%
\psfrag{VVVD}[cr][cr]{\PFGstyle $4$}%
\includegraphics[width=\textwidth]{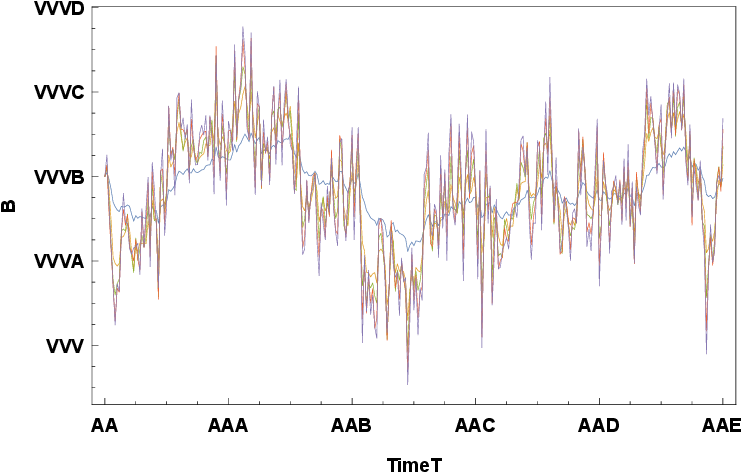}
\end{psfrags}
\end{minipage}%
\begin{minipage}{0.5\textwidth}
\tiny
\begin{psfrags}%
\psfrag{AAA}[tc][tc]{\PFGstyle $0.2$}%
\psfrag{AAB}[tc][tc]{\PFGstyle $0.4$}%
\psfrag{AAC}[tc][tc]{\PFGstyle $0.6$}%
\psfrag{AAD}[tc][tc]{\PFGstyle $0.8$}%
\psfrag{AAE}[tc][tc]{\PFGstyle $1.0$}%
\psfrag{AA}[tc][tc]{\PFGstyle $0.0$}%
\psfrag{B}[tc][tc]{\PFGstyle $W^{H,n}_t$}%
\psfrag{TimeT}[bc][Bc]{\PFGstyle $\text{time ($t$)}$}%
\psfrag{VVVA}[cr][cr]{\PFGstyle $ -5$}%
\psfrag{VVVB}[cr][cr]{\PFGstyle $ 0$}%
\psfrag{VVVC}[cr][cr]{\PFGstyle $ 5$}%
\psfrag{VVV}[cr][cr]{\PFGstyle $-10$}%
\psfrag{VVVD}[cr][cr]{\PFGstyle $10$}%
\includegraphics[width=\textwidth]{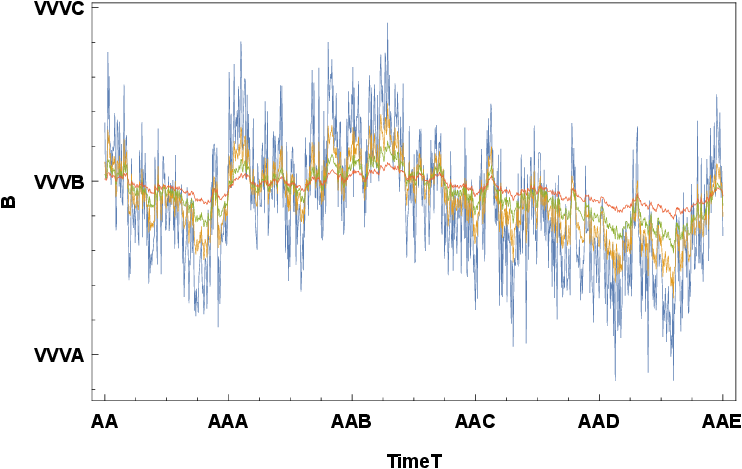}%
\end{psfrags}
\end{minipage}
\caption{%
Dependence of the approximations on the number $n$ of quadrature intervals and the Hurst index $H$. 
Left: varying the number 
$n\in\{2,5,10,20,40\}
=
\{\textcolor[rgb]{0.368417, 0.506779, 0.709798}{\blacksquare},
\textcolor[rgb]{0.880722, 0.611041, 0.142051}{\blacksquare},
\textcolor[rgb]{0.560181, 0.691569, 0.194885}{\blacksquare},
\textcolor[rgb]{0.922526, 0.385626, 0.209179}{\blacksquare},
\textcolor[rgb]{0.528488, 0.470624, 0.701351}{\blacksquare}\}$
of quadrature intervals with fixed parameters $H=0.1$, $m=5$. 
Right: 
varying the Hurst index 
$H\in\{0.1,0.2,0.3,0.4\}
=
\{\textcolor[rgb]{0.368417, 0.506779, 0.709798}{\blacksquare},
\textcolor[rgb]{0.880722, 0.611041, 0.142051}{\blacksquare},
\textcolor[rgb]{0.560181, 0.691569, 0.194885}{\blacksquare},
\textcolor[rgb]{0.922526, 0.385626, 0.209179}{\blacksquare}\}$ with fixed parameters $n=40$,  $m=5$.}
\label{fig:paths}
\end{figure}

\begin{figure}[h]%
\begin{minipage}{0.5\textwidth}
\tiny
\begin{psfrags}%
\psfrag{AAAAAA}[cr][cr]{\PFGstyle $10^{-2}$}%
\psfrag{AAAAAB}[cr][cr]{\PFGstyle $10^{-1}$}%
\psfrag{AAAAAC}[cr][cr]{\PFGstyle $10^{0}$}%
\psfrag{AAAAA}[cr][cr]{\PFGstyle $10^{-3}$}%
\psfrag{AAA}[tc][tc]{\PFGstyle $10^{1}$}%
\psfrag{AAB}[tc][tc]{\PFGstyle $10^{2}$}%
\psfrag{AAC}[tc][tc]{\PFGstyle $10^{3}$}%
\psfrag{AA}[tc][tc]{\PFGstyle $10^{0}$}%
\psfrag{B}[tc][tc]{\PFGstyle $\text{relative error ($e$)}$}%
\psfrag{NumberpOfInt}[bc][Bc]{\PFGstyle $\text{number of intervals ($n$)}$}%
\includegraphics[width=\textwidth]{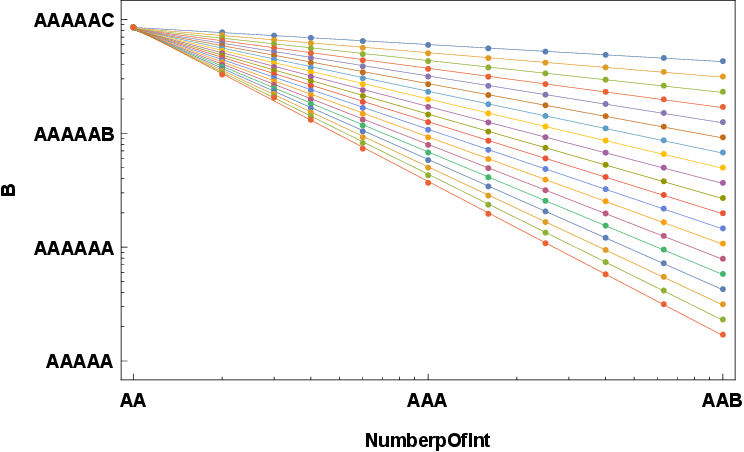}
\end{psfrags}
\end{minipage}%
\begin{minipage}{0.5\textwidth}
\tiny
\begin{psfrags}%
\psfrag{AAA}[tc][tc]{\PFGstyle $ 5$}%
\psfrag{AAB}[tc][tc]{\PFGstyle $10$}%
\psfrag{AAC}[tc][tc]{\PFGstyle $15$}%
\psfrag{AAD}[tc][tc]{\PFGstyle $20$}%
\psfrag{AA}[tc][tc]{\PFGstyle $ 0$}%
\psfrag{B}[tc][tc]{\PFGstyle $\text{rate ($r$)}$}%
\psfrag{InterpolOrd}[bc][Bc]{\PFGstyle $\text{interpolation order ($m$)}$}%
\psfrag{VVVA}[cr][cr]{\PFGstyle $0.2$}%
\psfrag{VVVB}[cr][cr]{\PFGstyle $0.4$}%
\psfrag{VVVC}[cr][cr]{\PFGstyle $0.6$}%
\psfrag{VVV}[cr][cr]{\PFGstyle $0.0$}%
\psfrag{VVVD}[cr][cr]{\PFGstyle $0.8$}%
\psfrag{VVVE}[cr][cr]{\PFGstyle $1.0$}%
\psfrag{VVVF}[cr][cr]{\PFGstyle $1.2$}%
\psfrag{VVVG}[cr][cr]{\PFGstyle $1.4$}%
\hspace{0.85em}\includegraphics[width=\textwidth-0.85em]{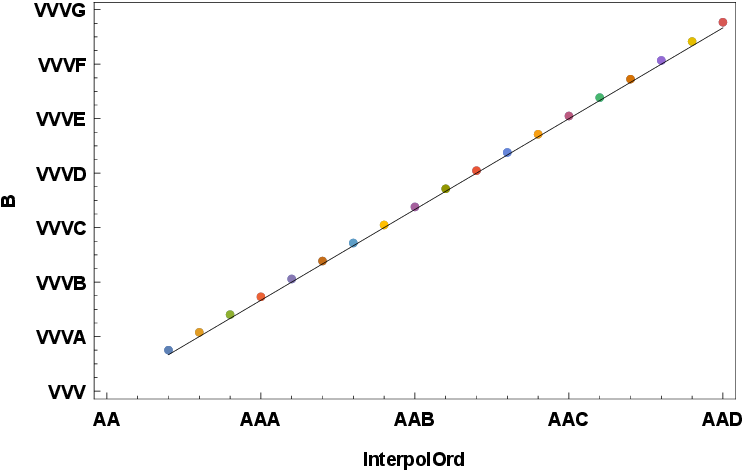}%
\end{psfrags}
\end{minipage}
\caption{
The upper bound $2Hm/3$ on the convergence rate established in Remark~\ref{rem:rate} for $m$-point interpolatory quadrature closely matches the numerically observed one (here: at $t=1$, computed analytically from the covariance functions of the Gaussian processes $W^H$ and $W^{H,n}$). 
Left: relative error $e=\|W^H_1-W^{H,n}_1\|_{L^2(\Omega)}/\|W^H_1\|_{L^2(\Omega)}$ for $m\in\{2,3,\dots,20\}=\{\textcolor[rgb]{0.368417, 0.506779, 0.709798}{\blacksquare},\textcolor[rgb]{0.880722, 0.611041, 0.142051}{\blacksquare},\dots,\textcolor[rgb]{0.8439466852489265, 0.3467106629502147, 0.3309221912517893}{\blacksquare}\}$ with $H=0.1$. 
Right: slopes of the lines in the left plot (dots) and predicted convergence rate (line).}
\label{fig:errors}
\end{figure}

The motivation of this article is to develop efficient Monte Carlo methods for fractional (or rough) volatility models \cite{gatheral2018volatility, bayer2016pricing, bennedsen2016decoupling, bayer2020regularity, horvath2017functional}, which have been introduced on the grounds of extensive empirical evidence \cite{gatheral2018volatility, bayer2016pricing,  bennedsen2016decoupling} and theoretical results \cite{alos2007short, fukasawa2011asymptotic, forde2017asymptotics, bayer2019short}.
Under our discretization, put prices in the rough Bergomi model converge at the same rate as the underlying fractional volatility process; see \autoref{thm:main}. 
By put-call parity, this extends to call prices if the the asset and volatility processes are driven by negatively correlated Brownian motions, as explained at the end of  \autoref{rem:putcall}.
A fully discrete Monte Carlo scheme for the rough Bergomi model can be obtained by discretizing the Ornstein--Uhlenbeck processes of \autoref{thm:main} in time. 
This can be done efficiently because the covariance matrix of the Ornstein--Uhlenbeck increments has low numerical rank if the time steps are small. 

To evaluate the computational complexity of our method, we consider the task of sampling a fractional process \smash{$(W^H_{i/k})_{i\in\{1,\dots,k\}}$} with Hurst index $H\in(0,1/2)$ at a temporal grid of $k$ equidistant time points.  
Our method has some additional parameters, which determine the spatial discretization of the integral representation, namely the number $n$ of spatial quadrature intervals and the order $m$ of the spatial quadrature. 
These are described in detail in \autoref{lem:dis}. 
On the above-mentioned task, our method achieves accuracy $n^{-r}$ at complexity $kn$ if the order of spatial quadrature is sufficiently high, i.e., if $m>3r/(2H)$ (see \autoref{rem:rate}).
Equivalently, accuracy $\epsilon$ can be achieved at complexity $k\epsilon^{-1/r}$, as stated in \autoref{tab:comparison}. 
Typically, one is interested in temporal grids of size $k=\epsilon^{-s}$ for some $s \in (0,\infty)$.
For instance, a value of $s$ slightly above $1/H$ guarantees that the piecewise constant interpolation of an $\epsilon$-accurate time-discrete approximation defines a continuous-time approximation of the same order of accuracy in the supremum norm. 
This is because the sample paths of the fractional process $W^H$ are nearly $H$-H\"older continuous. 
Under the assumption $k=\epsilon^{-s}$, \autoref{tab:comparison} shows that our method outperforms the methods Hosking and Dieker \cite{hosking1984modeling, dieker2004simulation} and Carmona, Coutin, and Montseny \cite{carmona2000approximation} but is outperformed by the hybrid scheme of Bennedsen, Lunde, and Pakkanen \cite{bennedsen2017hybrid} and by the circulant embedding method of Dietrich and Newsam \cite{dietrich1997fast}.
This can be verified by substituting $k=\epsilon^{-s}$ in Table~1.
Using exponentially converging quadrature rules such as Chebychev \cite{gass2016magic, gass2018chebyshev}, one could at best hope to reduce the complexity of our method from $k\epsilon^{-1/r}$ down to $k\log \epsilon^{-1}$. 
In the important special case $k=\epsilon^{-s}$ with $s=1/H$, this would result in exactly the same the complexity $\epsilon^{-1/H}\log\epsilon^{-1}$ as the hybrid scheme \cite{bennedsen2017hybrid} and the circulant embedding method \cite{dietrich1997fast}. 

\begin{table}[h]
\centering
\begin{tabular}{lccc}
\toprule
Method & Structure & Error & Complexity \\
\midrule
Cholesky & Static & 0 & $k^3$ \\
Hosking, Dieker \cite{hosking1984modeling, dieker2004simulation} & Recursive & 0 & $k^2$ \\
Dietrich, Newsam \cite{dietrich1997fast} & Static & 0 & $k\log k$ \\
Bennedsen, Lunde, Pakkanen \cite{bennedsen2017hybrid} & Recursive & $\epsilon=k^{-H}$ & $k \log k$ \\
Carmona, Coutin, Montseny \cite{carmona2000approximation} & Recursive & $\epsilon$ & $k\epsilon^{-3/(4H)}$ \\ 
This paper & Recursive & $\epsilon$ & $k\epsilon^{-1/r}$ for $r\in(0,\infty)$ \\
\bottomrule
\end{tabular}
\medskip
\caption{Complexity of several numerical methods for sampling a fractional process \smash{$(W^H_{i/k})_{i\in\{1,\dots,k\}}$} with Hurst index $H\in(0,1/2)$ at $k$ equidistant time points.}
\label{tab:comparison}
\end{table}

Several directions for future generalization and improvement come to mind. 
\autoref{thm:main} is proved by approximation in the Laplace domain, which implies convergence in the time domain by the continuity of the Laplace transform. 
As Volterra processes with Lipschitz drift and volatility coefficients depend continuously on the kernel in the $L^2$ norm, it would be interesting to check if similar convergence results hold also in this more general setting. 
The rate of convergence could potentially be improved using Chebychev quadrature, taking advantage of the real analyticity of the random field $Y_t(x)$ in the spatial variable $x$.
Finally, following \cite{bennedsen2017hybrid, mccrickerd2018turbocharging}, one could aim for more careful treatments of the singularity of the kernel near the diagonal and apply some variance reduction techniques. 

\section{Setting and notation}\label{sec:set}

We will frequently make the following assumptions.
Let $H\in (0,1/2)$,
let $\alpha=H+1/2$, 
let $\mu$ be the sigma-finite measure $x^{-\alpha}dx$ on the interval $(0,\infty)$,
let $p \in [1,\infty)$,
let $T \in (0,\infty)$,
let $(\Omega,\mathcal F,\mathbb P,(\mathcal F_t)_{t \in [0,T]})$ be a stochastic basis, 
and let $W,B\colon [0,T]\times\Omega\to\mathbb R$ be $(\mathcal F_t)_{t \in [0,T]}$-Brownian motions.

\section{Integral representation}\label{sec:int}

Recall from the introduction that Volterra Brownian motion $W^H$ can be lifted to a random field $Y_t(x)$ indexed by a temporal variable $t \in [0,\infty)$ and a spatial variable $x \in (0,\infty)$ \cite{carmona1998fractional, carmona2000approximation, muravlev2011representation, harms2019affine}.
The following lemma constructs a version of this random field which is continuous in the temporal variable and smooth in the spatial variable.
Moreover, it establishes bounds on the spatial derivatives and tails of the random field.
These bounds are needed for the subsequent error analysis in \autoref{sec:dis}.

The constants $m,\alpha,\beta,\gamma,\delta$ appearing in \autoref{lem:vol} are used consistently throughout the paper:
$m$ stands for the number of quadrature points in \autoref{def:qua} below, 
$\alpha=H+1/2$ denotes the Hurst index shifted by one half, 
$\beta$ describes spatial integrability of $\partial_x^m Y_t(x)$, 
$\gamma$ describes the integrability of the tail of $Y_t(x)$ as $x\to 0$, 
and $\delta$ describes the integrability of the tail of $Y_t(x)$ as $x\to\infty$.
The spaces of continuous, smooth, and integrable functions appearing in \autoref{lem:vol} carry their natural topologies and Borel sigma algebras; see \autoref{sec:aux}.

\begin{lemma}\label{lem:vol}
Assume the setting of \autoref{sec:set}.
\begin{enumerate}
\item\label{lem:vol0} There exists a measurable mapping
\begin{equation*}
Y \colon \Omega \to C([0,T],C^\infty((0,\infty),\mathbb R) \cap L^1((0,\infty),\mu)),
\end{equation*}
such that 
\begin{align*}
\forall t \in [0,\infty), \forall x \in (0,\infty): \quad \mathbb P\left[Y_t(x)=\frac{1}{\Gamma(\tfrac12-H)}\int_0^t e^{-(t-s)x} dW_s\right] = 1.
\end{align*}
\item\label{lem:vol1} Volterra Brownian motion is a linear functional of $Y$ in the sense that   
\begin{align*}
\forall t \in [0,T]: \quad \mathbb P\left[\int_0^\infty Y_t(x) \frac{dx}{x^\alpha}=\int_0^t (t-s)^{\alpha-1} dW_s\right] = 1.
\end{align*}

\item\label{lem:vol2} The following integrability conditions hold: for all $m \in \mathbb N_{>0}$, $\beta := m-1$, $\gamma := 1-\alpha$, and $\delta \in [0,\alpha-1/2)$, 
\begin{align*}
\left\|\sup_{t \in [0,T]}\sup_{x \in (0,\infty)}\left|x^\beta \partial_x^m Y_t(x)\right|\right\|_{L^p(\Omega)} &<\infty,
\\
\sup_{x_0\in [0,1]} x_0^{-\gamma} \left\| \sup_{t \in [0,T]} \left|\int_0^{x_0} Y_t(x)\frac{dx}{x^\alpha}\right|\right\|_{L^p(\Omega)} &< \infty, 
\\
\sup_{x_1\in [1,\infty)} x_1^{\delta} \left\| \sup_{t \in [0,T]} \left|\int_{x_1}^\infty Y_t(x)\frac{dx}{x^\alpha}\right|\right\|_{L^p(\Omega)} &< \infty.
\end{align*}
\end{enumerate}
\end{lemma}

\begin{proof}
\ref{lem:vol0} By \autoref{lem:continuity1, lem:continuity2}, the formula
\begin{equation*}
Y_t(x) := \frac{1}{\Gamma(\frac12-H)}\left(W_t - \int_0^t W_s x e^{-(t-s)x}ds\right), 
\qquad
t \in [0,T],
x \in (0,\infty), 
\end{equation*} 
defines a measurable map
\begin{equation*}
Y \colon \Omega \to C([0,T],C^\infty((0,\infty),\mathbb R) \cap L^1((0,\infty),\mu)).
\end{equation*}

\ref{lem:vol1} follows from the above and the stochastic Fubini theorem \cite{veraar2012stochastic}: for each $t \in [0,T]$, one has almost surely that
\begin{align*}
\int_0^\infty Y_t(x)\frac{dx}{x^\alpha}
&=
\frac{1}{\Gamma(\frac12-H)}\int_0^\infty \int_0^t e^{-(t-s)x} dW_s \frac{dx}{x^\alpha}
\\&=
\frac{1}{\Gamma(\frac12-H)} \int_0^t \int_0^\infty e^{-(t-s)x} \frac{dx}{x^\alpha} dW_s 
=
\int_0^t (t-s)^\alpha dW_s.
\end{align*}

\ref{lem:vol2}  
Let $C_1\in(0,\infty)$ be the constant in the maximal inequality for Ornstein--Uhlenbeck processes (see \autoref{lem:maximal}), i.e., 
\begin{align*}
\forall x \in (0,\infty):
\quad
\mathbb E\left[\sup_{t\in[0,T]}|Y_t(x)|\right]
\leq
C_1 \sqrt{\frac{\log(1+Tx)}{x}}.
\end{align*}
Recall that $\beta=m-1$, and define $C_2,C_3\in(0,\infty)$ as
\begin{align*}
C_2 &= 
\sup_{\substack{t \in (-\infty,0]\\x \in (0,\infty)}} |x^\beta\partial_x^m (x e^{tx})|
=
\sup_{\substack{t \in (-\infty,0]\\x \in (0,\infty)}} |x^{m-1}\partial_x^m \partial_t e^{tx}|
\\&=
\sup_{\substack{t \in (-\infty,0]\\x \in (0,\infty)}} |x^{m-1} \partial_t \partial_x^m e^{tx}|
=
\sup_{\substack{t \in (-\infty,0]\\x \in (0,\infty)}} |x^{m-1}\partial_t(t^m e^{tx})|
\\&=
\sup_{\substack{t \in (-\infty,0]\\x \in (0,\infty)}} \left|m(tx)^{m-1}+(tx)^m\right|e^{tx}
=
\sup_{y \in (-\infty,0]} \left|m y^{m-1} + y^m\right| e^{y}<\infty,
\\
C_3 &= \sup_{x \in (0,\infty)} x^{-(\alpha-\frac12-\delta)} \sqrt{\log(1+Tx)}<\infty.
\end{align*}
By 
the inequality $\log(1+Tx)\leq Tx$, one obtains the following three estimates:
\begin{align*}
\hspace{2em}&\hspace{-2em}
\mathbb E\left[\sup_{t \in [0,T]}\sup_{x \in (0,\infty)}\left|x^\beta \partial_x^m Y_t(x)\right|\right]
\\&= 
\mathbb E\left[\sup_{t \in [0,T]} \sup_{x \in (0,\infty)}\left| \int_0^t W_s x^\beta\partial_x^m (x e^{-(t-s)x}) ds \right| \right]
\\&\leq
C_2 T\, \mathbb E\left[\sup_{t \in [0,T]} |W_t| \right]<\infty,
\\
\hspace{2em}&\hspace{-2em}
\sup_{x_0\in [0,1]} x_0^{-\gamma} \mathbb E\left[ \sup_{t \in [0,T]} \left|\int_0^{x_0} Y_t(x)\frac{dx}{x^\alpha}\right|\right] 
\\&\leq
C_1 \sup_{x_0\in [0,1]} x_0^{-\gamma} \int_0^{x_0} \sqrt{\frac{\log(1+Tx)}{x}} \frac{dx}{x^\alpha}
\\&\leq
C_1 \sup_{x_0\in [0,1]} x_0^{-\gamma} \int_0^{x_0} \sqrt{T} \frac{dx}{x^\alpha}
= 
C_1 \sqrt{T} \gamma^{-1}
< \infty, 
\\
\hspace{2em}&\hspace{-2em}
\sup_{x_1\in [1,\infty)} x_1^{\delta} \mathbb E\left[ \sup_{t \in [0,T]} \left|\int_{x_1}^\infty Y_t(x)\frac{dx}{x^\alpha}\right|\right] 
\\&\leq
C_1 \sup_{x_1\in [1,\infty)} x_1^{\delta} \int_{x_1}^\infty \sqrt{\frac{\log(1+Tx)}{x}} \frac{dx}{x^\alpha}
\\&\leq
C_1 C_3 \sup_{x_1\in [1,\infty)} x_1^{\delta} \int_{x_1}^\infty x^{-1-\delta} dx
=
C_1 C_3 \delta^{-1} 
< \infty.
\end{align*}
This shows \ref{lem:vol2} for $p=1$. 
The generalization to $p \in [1,\infty)$ is immediate because the $L^p$ norms of a Banach-valued Gaussian random variable are mutually equivalent thanks to the Kahane--Khintchine inequality \cite[Theorem~V.5.3]{vakhania1987probability} applied to the Karhunen--Lo\`eve expansion \cite[Theorem~V.5.7]{vakhania1987probability}.
\end{proof}

\section{Discretization}\label{sec:dis}

In this section, the measure $\mu$ in the integral representation of Volterra Brownian motion is approximated by a weighted sum of Dirac measures.
More specifically, for each $n\in\mathbb N$, the positive half line is truncated to a finite interval $[\xi_{n,0},\xi_{n,n}]$. 
This interval is then split into subintervals by a geometric sequence $(\xi_{n,i})_{i\in\{1,\dots,n\}}$, 
and on each  subinterval $[\xi_{n,i},\xi_{n,i+1}]$ the measure $\mu$ is approximated by an $m$-point interpolatory quadrature rule $\mu_{n,i}$ such as e.g.\@ the Gauss rule.
Classical error analysis for interpolatory quadrature rules (see e.g.~\cite{brass2011quadrature}) then yields the desired convergence result. 

\begin{definition}\label{def:qua}
Let $a,b \in \mathbb R$ satisfy $a<b$, let $w\colon [a,b]\to [0,\infty)$ be a continuous function such that $\int_a^b w(x)dx>0$, and let $m \in \mathbb N_{>0}$. Then a measure $\mu$ on $[a,b]$ is called a non-negative $m$-point \emph{interpolatory quadrature rule} on $[a,b]$ with respect to the weight function $w$ if there are grid points $x_1,\dots,x_m \in [a,b]$ and weights $w_1,\dots,w_m \in [0,\infty)$ such that $\mu = \sum_{j=1}^m w_j \delta_{x_j}$ and 
\begin{equation*}
\forall k \in \{0,\dots,m-1\}:
\quad
\int_a^b x^k w(x) \mu(dx) = \int_a^b x^k w(x) dx.
\end{equation*}
\end{definition}

The following lemma discretizes the integral representation of Volterra Brownian motion using interpolatory quadrature rules and bounds the discretization error. 
The assumptions of the lemma are satisfied thanks to the bounds of \autoref{lem:vol}, where the same constants $\alpha,\beta,\gamma,\delta,m$ are used.

\begin{lemma}
\label{lem:dis}
Assume the setting of \autoref{sec:set},
let $m \in \mathbb N_{>0}$ and $\alpha,\beta,\gamma,\delta \in (0,\infty)$ satisfy $1-\alpha-\beta+m > 0$, 
let 
\begin{equation*}
Y\colon \Omega\to C([0,T], C^m((0,\infty))\cap L^1((0,\infty),\mu))
\end{equation*}
be a measurable function which satisfies the integrability conditions
\begin{align*}
\left\|\sup_{t \in [0,T]}\sup_{x \in (0,\infty)}\left|x^\beta \partial_x^m Y_t(x)\right|\right\|_{L^p(\Omega)} &< \infty,
\\
\limsup_{x_0\downarrow 0} x_0^{-\gamma} \left\| \sup_{t \in [0,T]} \left|\int_0^{x_0} Y_t(x)x^{-\alpha}dx\right|\right\|_{L^p(\Omega)} &< \infty, 
\\
\limsup_{x_1\uparrow \infty} x_1^{\delta} \left\| \sup_{t \in [0,T]} \left|\int_{x_1}^\infty Y_t(x)x^{-\alpha}dx\right| \right\|_{L^p(\Omega)} &< \infty, 
\end{align*}
let $r \in (0,\delta m/(1-\alpha-\beta+\delta+m))$, 
for each $n \in \mathbb N$ and $i \in \{0,\dots,n-1\}$ let
\begin{equation*}
\xi_{n,0}=n^{-r/\gamma}, 
\qquad
\xi_{n,n}=n^{r/\delta}, 
\qquad
\xi_{n,i}=\xi_{n,0}(\xi_{n,n}/\xi_{n,0})^{i/n},
\end{equation*}
let $\mu_{n,i}$ be a non-negative $m$-point interpolatory quadrature rule on $[\xi_{n,i},\xi_{n,i+1}]$ with respect to the weight function $x\mapsto x^{-\alpha}$, 
and let $\mu_n=\sum_{i=0}^{n-1}\mu_{n,i}$. 
Then 
\begin{equation*}
\sup_{n\in\mathbb N} n^r \left\|\sup_{t \in [0,T]} \left|\int_0^\infty Y_t(x) x^{-\alpha} (\mu_n(dx)-dx)\right|\right\|_{L^p(\Omega)}<\infty.
\end{equation*}
\end{lemma}

\begin{proof} We define the constants
\begin{align*}
\eta &= \left(\frac{1}{r}-\frac{1-\alpha-\beta+m+\delta}{\delta m}\middle)\middle\slash\middle(\frac{1}{\gamma}+\frac{1}{\delta}\right) \in (0,\infty),
\\
C_1
&=
\frac{\pi^m}{m!2^m} \left\|\sup_{t \in [0,T]}\sup_{x \in (0,\infty)}|x^\beta Y^{(m)}_t(x)|\right\|_{L^p(\Omega)}
\in (0,\infty),
\\
C_2
&= 
\sup_{\lambda \in (1,\infty)} \frac{\lambda-1}{\lambda^{1-\alpha-\beta+m}-1}
\in
(0,1/(1-\alpha-\beta+m)],
\\
C_3 
&= 
\sup_{\xi \in [1,\infty)} \sup_{n \in [\log \xi,\infty)} \left(\xi^{1/n}-1\right) n \xi^{-\eta}
\in (0,\infty),
\\
C_4
&=
\min \left\{n \in \mathbb N; n\geq \log(\xi_{n,n}/\xi_{n,0})=\left(\frac{r}{\gamma}+\frac{r}{\delta}\right)\log(n)\right\}  
< \infty,
\end{align*}
where the upper bound on $C_2$ follows from Bernoulli's inequality
\begin{equation*}
\forall \lambda \in [0,\infty):
\quad
\lambda^{1-\alpha-\beta+m}=(1+(\lambda-1))^{1-\alpha-\beta+m} \geq 1 +(1-\alpha-\beta+m)(\lambda-1),
\end{equation*}
the finiteness of $C_3$ follows from the inequality
\begin{align*}
\forall \xi \in [1,\infty), \forall n \in [\log(\xi),\infty):
\quad
\xi^{1/n}-1
&=
\exp(\log(\xi)/n)-1
\leq 
e \log(\xi)/n,
\end{align*}
and the finiteness of $C_4$ follows from the fact that $n^{-1}\log(n)$ tends to zero as $n\to\infty$.
Recall that the measures $\mu_{n,i}$ are by assumption non-negative $m$-point interpolatory quadrature rules.  
Therefore, the corresponding quadrature error can be expressed as follows \cite[Theorem~4.2.3]{brass2011quadrature}: for each $t \in [0,T]$, $n \in \mathbb N$, and $i \in \{0,\dots,n-1\}$, one has
\begin{align*}
\hspace{2em}&\hspace{-2em}
\int_{\xi_{n,i}}^{\xi_{n,i+1}} Y_t(x) x^{-\alpha} (\mu_n(dx)-dx)
= 
\int_{\xi_{n,i}}^{\xi_{n,i+1}} \partial_x^m Y_t(x) K_{n,i}(x)dx,
\end{align*}
where the Peano kernel $K_{n,i} \colon [\xi_{n,i},\xi_{n,i+1}]\to\mathbb R$ is a measurable function which satisfies \cite[Theorem~5.7.1]{brass2011quadrature} 
\begin{align*}
\sup_{x \in [\xi_{n,i},\xi_{n,i+1}]} |K_{n,i}(x)|
\leq 
\frac{\pi^m}{m!} \left(\frac{\xi_{n,i+1}-\xi_{n,i}}{2}\right)^m \sup_{x \in [\xi_{n,i},\xi_{n,i+1}]} x^{-\alpha}. 
\end{align*}
Thus, one has for each $n \in \mathbb N$ that
\begin{align*}\tag{$*$}\label{equ:dis*}
\hspace{2em}&\hspace{-2em}
\left\|\sup_{t \in [0,T]}\left|\int_{\xi_{n,0}}^{\xi_{n,n}} Y_t(x) x^{-\alpha} (\mu_n(dx)-dx)\right|\right\|_{L^p(\Omega)}
\\&\leq
\sum_{i=0}^{n-1}\left\|\sup_{t \in [0,T]}\left|\int_{\xi_{n,i}}^{\xi_{n,i+1}} Y_t(x) K_{n,i}(x) dx\right|\right\|_{L^p(\Omega)}
\\&\leq
\sum_{i=0}^{n-1} \frac{\pi^m}{m!2^m} \left\|\sup_{\substack{t \in [0,T]\\x \in [\xi_{n,i},\xi_{n,i+1}]}}|x^\beta Y^{(m)}_t(x)|\right\|_{L^p(\Omega)}\!\! \xi_{n,i}^{-\alpha-\beta} (\xi_{n,i+1}-\xi_{n,i})^{m+1}
\\&\leq
C_1 \sum_{i=0}^{n-1} \xi_{n,i}^{-\alpha-\beta} (\xi_{n,i+1}-\xi_{n,i})^{m+1}. 
\end{align*}
This can be expressed as a geometric series: letting $\lambda_n = (\xi_{n,n}/\xi_{n,0})^{1/n}$, one has for each $n \in \mathbb N$ that
\begin{align*}
\eqref{equ:dis*} &=
C_1 \xi_{n,0}^{1-\alpha-\beta+m} (\lambda_n-1)^{m+1} \sum_{i=0}^{n-1} \lambda_n^{i(1-\alpha-\beta+m)} 
\\&=
C_1 \xi_{n,0}^{1-\alpha-\beta+m} (\lambda_n-1)^{m+1} \frac{\lambda_n^{n(1-\alpha-\beta+m)}-1}{\lambda_n^{1-\alpha-\beta+m}-1}
\\&=
C_1 (\lambda_n-1)^{m+1} \frac{\xi_{n,n}^{1-\alpha-\beta+m}-\xi_{n,0}^{1-\alpha-\beta+m}}{\lambda_n^{1-\alpha-\beta+m}-1}.
\end{align*}
Absorbing the denominator into one of the factors $(\lambda_n-1)$ and discarding the term $\xi_{n,0}$ yields for each $n \in \mathbb N$ that
\begin{align*}
\eqref{equ:dis*}&\leq
C_1 C_2(\lambda_n-1)^m \xi_{n,n}^{1-\alpha-\beta+m}
=
C_1 C_2((\xi_{n,n}/\xi_{n,0})^{1/n}-1)^m \xi_{n,n}^{1-\alpha-\beta+m}.
\end{align*}
For each $n \in \mathbb N \cap [C_4,\infty)$, this can be estimated by
\begin{align*}
\eqref{equ:dis*}&\leq
C_1 C_2 C_3^m n^{-m} (\xi_{n,n}/\xi_{n,0})^{\eta m} \xi_{n,n}^{1-\alpha-\beta+m}
\\&=
C_1 C_2 C_3^m n^{-m+\eta m r (1/\gamma+1/\delta)+(1-\alpha-\beta+m)r/\delta}
=
C_1 C_2 C_3^m n^{-r}.
\end{align*}
Therefore, noting that $n^r = \xi_{n,0}^{-\gamma}=\xi_{n,n}^{\delta}$, one has
\begin{align*}
\hspace{2em}&\hspace{-2em}
\limsup_{n\to\infty} n^r \mathbb E\left[\sup_{t \in [0,T]} \left|\int_0^\infty Y_t(x) x^{-\alpha} (\mu_n(dx)-dx)\right|\right]
\\&\leq
\limsup_{n\to\infty} \xi_{n,0}^{-\gamma}\ \mathbb E\left[\sup_{t \in [0,T]} \left|\int_{(0,\xi_{n,0}]} Y_t(x) x^{-\alpha} dx\right|\right]
\\&\qquad+
\limsup_{n\to\infty} \xi_{n,n}^{\delta}\ \mathbb E\left[\sup_{t \in [0,T]} \left|\int_{[\xi_{n,n},\infty)} Y_t(x) x^{-\alpha} dx\right|\right]
\\&\qquad+
\sup_{n\in\mathbb N} n^r \mathbb E\left[\sup_{t \in [0,T]}\left|\int_{\xi_{n,0}}^{\xi_{n,n}} Y_t(x) x^{-\alpha} (\mu_n(dx)-dx)\right|\right]
<\infty.
\qedhere
\end{align*}
\end{proof}

\begin{remark}
The choice of the quadrature rule in \autoref{lem:dis} is admittedly somewhat arbitrary but produces good results. 
The use of the geometric grid $\xi_{n,i}$ goes back to \cite{carmona2000approximation} and simplifies the error analysis compared to more complex subdivisions which distribute the error more equally. 
It would be interesting to explore if the holomorphicity of $x\mapsto Y_t(x)$ permits the use of quadrature rules with exponential convergence rates such as Chebychev quadrature; see the discussion in \autoref{sec:int}.
\end{remark}

\section{Rough Bergomi model}

The following lemma establishes that prices of put options in the rough Bergomi model converge at the same rate as the approximated Volterra processes. 
This holds not only for the Ornstein--Uhlenbeck approximations of \autoref{lem:dis}, but more generally for any approximation of the log-volatility in the $L^2([0,T]\times\Omega)$ norm.
Below, the space of real-valued Lipschitz functions $f\colon\mathbb R\to\mathbb R$ is denoted by $\operatorname{Lip}(\mathbb R)$ and endowed with the norm $\|f\|_{\operatorname{Lip}(\mathbb R)}=|f(0)|+\sup_{x\neq y}|f(y)-f(x)| |y-x|^{-1}$.

\begin{lemma}\label{lem:bergomi}
Assume the setting of \autoref{sec:set},
let $V,\smashedtilde V,S,\smashedtilde S\colon [0,T]\times\Omega\to\mathbb R$ be continuous stochastic processes with $V_0=\smashedtilde V_0=0$ and
\begin{equation*}
\forall t \in [0,T]: 
\quad
S_t = 1 + \int_0^t S_s \exp(V_s)dW_s, 
\quad
\smashedtilde S_t = 1 + \int_0^t \smashedtilde S_s \exp(\smashedtilde V_s)dW_s,
\end{equation*}
and let $f\colon(0,\infty)\to\mathbb R$ be a measurable function such that $f \circ \exp \in \operatorname{Lip}(\mathbb R)$. 
Then 
\begin{multline*}
\big|\mathbb E[f(S_T)]-\mathbb E[f(\smashedtilde S_T)]\big|
\leq 
\|f\circ\exp\|_{\operatorname{Lip}(\mathbb R)}\big(\sqrt{T}+6\big)
\\
\times\left\|\exp(2|V|)+\exp(2|\smashedtilde V|)\right\|_{L^2(\Omega,C([0,T]))}\|V-\smashedtilde V\|_{L^2([0,T]\times\Omega)}.
\end{multline*}
\end{lemma}

\begin{proof}
It is sufficient to control the log prices in $L^1$ because
\begin{equation*}
\big|\mathbb E[f(S_T)]-\mathbb E[f(\smashedtilde S_T)]\big|
\leq
\|f\circ\exp\|_{\operatorname{Lip}(\mathbb R)} \|\log(S_T)-\log(\smashedtilde S_T)\|_{L^1(\Omega)}.
\end{equation*}
The basic inequality
\begin{equation*}
\forall x,y \in \mathbb R:
\quad
\big|\exp(x)-\exp(y)\big|
\leq 
\big(\exp(x)+\exp(y)\big) |x-y|
\end{equation*}
and the Burkholder--Davis--Gundy inequality \cite[Theorem~1.2]{beiglbock2015pathwise} imply that
\begin{align*}
\hspace{2em}&\hspace{-2em}
\|\log(S_T)-\log(\smashedtilde S_T)\|_{L^1(\Omega)}
\\&=
\left\|-\frac12\int_0^T\big(\exp(2V_t)-\exp(2\smashedtilde V_t)\big)dt+\int_0^T\big(\exp(V_t)-\exp(\smashedtilde V_t)\big)dW_t\right\|_{L^1(\Omega)}
\\&\leq
\left\|\frac12\int_0^T\big(\exp(2V_t)+\exp(2\smashedtilde V_t)\big)(2V_t-2\smashedtilde V_t)dt\right\|_{L^1(\Omega)}
\\&\qquad
+6\left\|\sqrt{\int_0^T\big(\exp(V_t)+\exp(\smashedtilde V_t)\big)^2(V_t-\smashedtilde V_t)^2dt}\right\|_{L^1(\Omega)}
\\&\leq
\left\|\exp(2|V|)+\exp(2|\smashedtilde V|)\right\|_{L^2(\Omega,C([0,T]))}
\\&\qquad\times
\left(\left\|\int_0^T(V_t-\smashedtilde V_t)dt\right\|_{L^2(\Omega)}+6\left\|\sqrt{\int_0^T(V_t-\smashedtilde V_t)^2dt}\right\|_{L^2(\Omega)}\right)
\\&\leq
\left\|\exp(2|V|)+\exp(2|\smashedtilde V|)\right\|_{L^2(\Omega,C([0,T]))}
\big(\sqrt{T}+6\big) \left\|V-\smashedtilde V\right\|_{L^2([0,T]\times\Omega)}.
\qedhere
\end{align*}
\end{proof}

\begin{remark}\label{rem:putcall}
For each $K \in (0,\infty)$ the put-option payoff 
\begin{equation*}
f\colon (0,\infty) \to \mathbb R, \qquad x \mapsto (K-x)_+,
\end{equation*}
satisfies the assumption of \autoref{lem:bergomi} that $f\circ\exp\in\operatorname{Lip}(\mathbb R)$ because
\begin{equation*}
\sup_{\substack{x,y\in\mathbb R\\x\neq y}}\frac{|f(e^y)-f(e^x)|}{|y-x|}\leq e^K<\infty.
\end{equation*}
The call-option payoff does not have this property, but the prices of call options can be obtained by put-call parity if $W$ and $B$ are negatively correlated because this implies that $S$ is a martingale \cite{gassiat2019martingale}.
\end{remark}

\section{Main result}\label{sec:main}

The following theorem combines the analyses of Lemmas~\ref{lem:vol}--\ref{lem:bergomi} to show that Volterra Brownian motion can be approximated numerically at arbitrarily high polynomial convergence rates $r$.
The same convergence rate $r$ is inherited by the associated put prices in the rough Bergomi model.

\begin{theorem}\label{thm:main}
Assume the setting of \autoref{sec:set}.
For any given $r \in (0,\infty)$,
the following statements hold:
\begin{enumerate}
\item\label{thm:main1} 
Volterra Brownian motion can be approximated at rate $n^{-r}$ by a sum of $n$ Ornstein--Uhlenbeck processes in the following sense: 
for each $n\in\mathbb N$ there are speeds of mean reversion $x_{n,i}\in(0,\infty)$ and weights $w_{n,i}\in(0,\infty)$, $1\leq i\leq n$, such that the continuous versions $W^H$ and $W^{H,n}$ of the stochastic integrals
\begin{align*}
W^H_t&:=\int_0^t (t-s)^{H-1/2} dW_s, 
&
W^{H,n}_t&:=\sum_{i=1}^n w_{n,i}\int_0^t e^{-(t-s)x_{n,i}} dW_s,
&&
t \in [0,T].
\end{align*}
satisfy
\begin{equation*}
\sup_{n \in \mathbb N} n^r \left\| W^H-W^{H,n}\right\|_{L^p(\Omega,C([0,T],\mathbb R))}<\infty.
\end{equation*}

\item\label{thm:main2} 
Under the above approximation, put prices in the rough Bergomi model converge at rate $n^{-r}$ in the following sense: 
the processes $S$ and $S^n$ defined for all $t \in [0,T]$ and $n \in \mathbb N$ by 
\begin{align*}
S_t &:= 1 + \int_0^t S_s \exp\big(W^H_s-\tfrac12\mathbb E[(W^H_s)^2]\big) dB_s, 
\\ 
S^{n}_t &:= 1 + \int_0^t S^n_s \exp\big(W^{H,n}_s-\tfrac12\mathbb E[(W^{H,n}_s)^2]\big) dB_s, 
\end{align*}
satisfy for all strikes $K \in [0,\infty)$ that
\begin{equation*}
\sup_{n \in \mathbb N} n^r \left| \mathbb E\left[(K-S_T)_+\right]-\mathbb E\left[(K-S^n_T)_+\right]\right|<\infty.
\end{equation*}
\end{enumerate}
\end{theorem}

\begin{proof}
\ref{thm:main1} follows from the integral representation in \autoref{lem:vol} and its discretization in \autoref{lem:dis}. 
More precisely, the $m$-point quadrature rule in \autoref{lem:dis} converges at any rate $r<\delta m/(1-\alpha-\beta+\delta+m)=2Hm/3$, where the parameters $\alpha=H+1/2$, $\beta=m-1$, $\gamma=1/2-H$, and $\delta=H$ are as in \autoref{lem:vol}. 
The speeds of mean reversion $x_{n,i}$ and weights $w_{n,i}$ are determined by the relation $\mu_n=\sum_i w_{n,i}\delta_{x_{n,i}}$, where $\mu_n$ is as in \autoref{lem:dis}.
Moreover, \ref{thm:main2} follows from \ref{thm:main1} and \autoref{lem:bergomi}.
\end{proof}

\begin{remark}
\label{rem:rate}
The proof of \autoref{thm:main} shows that $m$-point interpolatory quadrature on $n$ suitably chosen spatial quadrature intervals leads to a discretization error of order $n^{-r}$ for all $r \in (0,2Hm/3)$.
\end{remark}

\appendix

\section{Auxiliary results}
\label{sec:aux}

The space $C([0,T],\mathbb R)$ of continuous real-valued functions on an interval $[0,T]$ is Banach with the supremum norm.
Moreover, the space $C^\infty((0,\infty),\mathbb R)$ of smooth real-valued functions on $(0,\infty)$ is locally convex with the family of seminorms $f\mapsto\sup_{x\in K}|\partial_x^k f(x)|$, where $K$ runs through the compact subsets of $(0,\infty)$ and $k$ through the natural numbers. 
Similarly, the space $C([0,T],C^\infty((0,\infty),\mathbb R))$ is locally convex with the family of seminorms $f\mapsto\sup_{t \in [0,T]}\sup_{x \in K}|\partial_x^k f(t)(x)|$ for $K$ and $k$ as before.

\begin{lemma}
\label{lem:continuity1}
Assume the setting of \autoref{sec:set}.
Then the following function is continuous:
\begin{equation*}
C([0,T],\mathbb R) \ni w \mapsto \Bigg((t,x)\mapsto w_t - \int_0^t w_s x e^{-(t-s)x}ds\Bigg) \in C([0,T],C^\infty((0,\infty),\mathbb R)).
\end{equation*}
\end{lemma}

\begin{proof}
It is sufficient to show for each $k \in \mathbb N$ and each compact $K \subset (0,\infty)$ that the following mapping is continuous: 
\begin{equation*}
C([0,T],\mathbb R) \ni w \mapsto \Bigg((t,x)\mapsto \partial_x^k w_t - \partial_x^k\int_0^t w_s x e^{-(t-s)x}ds\Bigg) \in C([0,T],C(K,\mathbb R)).
\end{equation*}
This is obvious because this is a bounded linear map between Banach spaces. 
\end{proof}

The following maximal inequality for Ornstein--Uhlenbeck processes has been shown by \textcite[Theorem~2.5 and Remark~2.6]{graversen2000maximal}. 

\begin{lemma}
\label{lem:maximal}
Assume the setting of \autoref{sec:set}.
For each $x \in (0,\infty)$, let $Y(x)\colon\Omega\to C([0,T],\mathbb R)$ be a measurable map such that
\begin{align*}
\forall t \in [0,T], \forall x \in (0,\infty):
\quad
\mathbb P\left[Y_t(x) = \frac{1}{\Gamma(\frac12-H)}\int_0^t e^{-(t-s)x}dW_s\right]=1.
\end{align*}
Then there exists a universal constant $C_1 \in (0,2)$ such that the following maximal inequality holds: 
\begin{align*}
\forall x \in (0,\infty):
\quad
\mathbb E\left[\sup_{t\in[0,T]}|Y_t(x)|\right]
\leq
C_1 \sqrt{\frac{\log(1+Tx)}{x}}.
\end{align*}
\end{lemma}

The following result has been shown in \cite[Theorem~2.11]{harms2019affine}. 
We reproduce the argument here and give a simpler proof of measurability.
Recall from \autoref{sec:set} that $\mu=x^{-\alpha}dx$ is a sigma-finite measure on $(0,\infty)$ and that, accordingly, the space $L^1((0,\infty),\mu)$ of $\mu$-integrable functions is a separable Banach space. 
Its intersection with the locally convex space $C((0,\infty),\mathbb R)$ is again locally convex with the union of the corresponding families of seminorms.

\begin{lemma}
\label{lem:continuity2}
Assume the setting of \autoref{sec:set}, 
and let $Y\colon\Omega\to C([0,T],C((0,\infty),\mathbb R))$ be a measurable map such that
\begin{align*}
\forall t \in [0,T], \forall x \in (0,\infty):
\quad
\mathbb P\left[Y_t(x) = \frac{1}{\Gamma(\frac12-H)}\int_0^t e^{-(t-s)x}dW_s\right]=1.
\end{align*}
Then $Y$ almost surely takes values in the space $C([0,T],L^1((0,\infty),\mu))$ and is measurable as a map
$$Y\colon \Omega\to C([0,T],C((0,\infty),\mathbb R)\cap L^1((0,\infty),\mu)).$$
\end{lemma}

\begin{proof}
The expression
\begin{align*}
\mathbb E\left[\int_0^\infty \sup_{t \in [0,T]} |Y_t(x)|\frac{dx}{x^\alpha}\right]
=
\int_0^\infty \mathbb E\left[\sup_{t \in [0,T]} |Y_t(x)|\right]\frac{dx}{x^\alpha}
\end{align*}
is well-defined thanks to the continuity in $t$ of $Y_t(x)$, and is finite thanks to \autoref{lem:maximal}. 
Thus, the dominated convergence theorem implies that $Y$ has continuous sample paths in $L^1((0,\infty),\mu)$.
It remains to show that $Y\colon\Omega\to C([0,T],L^1((0,\infty),\mu))$ is measurable. 
As the Borel sigma algebra on $C([0,T],L^1((0,\infty),\mu))$ is generated by point evaluations at $t \in [0,T]$ \cite[Lemma~4.53]{aliprantis2006infinite}, it suffices to show for each $t \in [0,T]$ that $Y_t\colon\Omega\to L^1((0,\infty))$ is measurable.
Moreover, by Pettis' measurability theorem \cite[Proposition~1.1.1]{hytoenen2016analysis} and the separability of $L^1((0,\infty),\mu)$, it suffices to show that $Y_t$ is weakly measurable, i.e., that $\int_0^\infty Y_t(x) f(x) \mu(dx)\colon\Omega\to\mathbb R$ is measurable for each $f \in L^\infty((0,\infty),\mu)$. 
This follows by approximation 
\begin{align*}
\int_0^\infty Y_t(x) f(x) \mu(dx) 
=
\lim_{n\to\infty}\lim_{m\to\infty} \int_{1/n}^n Y_t(x) \mu_{n,m}(dx),
\end{align*}
where for each $n \in \mathbb N$, $(\mu_{n,m})_{m\in\mathbb N}$ is a sequence of atomic signed measures on the interval $[1/n,n]$ which converges weakly to the signed measure $f\mu$ on the same interval.
\end{proof}

\printbibliography

\end{document}